\documentclass[a4paper,UKenglish]{article}

\usepackage{authblk}
\usepackage{array}
\usepackage{graphicx}
\usepackage{arydshln}
\usepackage[noend]{algpseudocode}
\usepackage{algorithmicx}
\usepackage[ruled]{algorithm}
\usepackage{algpseudocode}

\algnewcommand{\LineComment}[1]{\State \(\triangleright\) #1}

\usepackage{tikz}
\usepackage{pgfplotstable}
\pgfplotsset{width=7cm,compat=1.3}
\usetikzlibrary{trees}
\usetikzlibrary{decorations.pathmorphing, decorations.pathreplacing}
\usetikzlibrary{decorations.markings}
\usetikzlibrary{decorations.pathmorphing,shapes}
\usetikzlibrary{calc,decorations.pathmorphing,shapes}
\usetikzlibrary{positioning, fit, backgrounds, calc, shadows,arrows, shapes, trees, mindmap, decorations.text}

  \usepackage{forest}
  \usepackage{xspace}
  \usepackage{enumerate}

  \usepackage{multirow}
  \usepackage{todonotes}
  \usepackage{graphics,adjustbox}

  \usepackage{epsfig}
  \usepackage{url}

  \usepackage{verbatim}
  \usetikzlibrary{trees, arrows, shapes}
\usepackage{xcolor}
\colorlet{grey}{gray}
\colorlet{ColorD1}{grey!10}
\colorlet{ColorD2}{grey!20}
\colorlet{ColorD3}{grey!30}
\colorlet{ColorD4}{grey!40}
\colorlet{ColorD5}{grey!50}
\colorlet{ColorD6}{grey!60}
\colorlet{ColorD7}{grey!70}
\colorlet{ColorD8}{grey!80!}

\tikzstyle{vertex}=[circle,draw=black, fill=black!25,minimum size=20pt,inner sep=0pt, text centered, text width=1.5em]
\tikzstyle{edge} = [draw,thick,-]

\tikzset{
  treenode/.style = {align=center, inner sep=2pt, text centered,
    font=\sffamily},
  arn_r/.style = {treenode, circle, black, font=\sffamily\bfseries, draw=black,
    text width=1.5em},
    arn_t/.style = {treenode, circle, black, thick, double, font=\sffamily\bfseries, draw=black,
    text width=1.5em},
  every edge/.append style={anchor=south,auto=falseanchor=south,auto=false,font=3.5 em},
}

\newcommand{\hook}[2]{\mathcal{H}_{#1}^{#2}}

\usepackage{amssymb}
\usepackage{amsmath}
\usepackage{amsthm}
\newtheorem{theorem}{Theorem}

\newtheorem{observation}[theorem]{Observation}
\newtheorem{proposition}[theorem]{Proposition}
\newtheorem{lemma}[theorem]{Lemma}
\newtheorem{example}{Example}

\theoremstyle{definition}
\newtheorem{mydef}{Definition}[section]
\newtheorem{remark}{Remark}

\def\dd{\mathinner{.\,.}}

\newcommand{\Oh}{\mathcal{O}}
\newcommand{\LUF}{\mathsf{LUF}}
\newcommand{\HOOK}{\mathsf{HOOK}}
\newcommand{\LSFr}{\mathsf{LSF}_{r}}
\newcommand{\LSFl}{\mathsf{LSF}_\ell}
\newcommand{\Stack}{\mathcal{S}}
\newcommand{\Twin}{\mathcal{T}}

 \newcommand{\defproblem}[3]{
  \vspace{3mm}
\noindent\fbox{
  \begin{minipage}{0.96\textwidth}
  #1\\
  {\bf{Input:}} #2  \\
  {\bf{Output:}} #3
  \end{minipage}
  }
  \vspace{3mm}
}

\title{Longest Unbordered Factor in Quasilinear Time}

\author[1]{Tomasz Kociumaka\thanks{kociumaka@mimuw.edu.pl}}
\author[2]{Ritu Kundu\thanks{ritu.kundu@kcl.ac.uk}}
\author[2]{Manal Mohamed\thanks{manal.mohamed@kcl.ac.uk}}
\author[2]{Solon P. Pissis\thanks{solon.pissis@kcl.ac.uk}}
\affil[1]{Institute of Informatics, University of Warsaw, Warsaw, Poland}
\affil[2]{Department of Informatics, King's College London, London, UK }

\begin{document}

\maketitle

\begin{abstract}
A {\em border} $u$ of a word $w$ is a proper factor of $w$ occurring both as a prefix and as a suffix. The {\em maximal unbordered factor} of $w$ is the longest factor of $w$ which does not have a border. Here an $\Oh(n\log n)$-time with high probability (or $\Oh(n \log n \log^2 \log n)$-time deterministic) algorithm to compute the \textit{Longest Unbordered Factor Array} of $w$ for general alphabets is presented, where $n$ is the length of $w$. This array specifies the length of the maximal unbordered factor starting at each position of $w$. This is a major improvement on the running time of the currently best worst-case algorithm working in $\mathcal{O}(n^{1.5})$ time for integer alphabets [Gawrychowski et al., 2015]. 

\end{abstract}

\section{Introduction}\label{sect:intro}

There are two central properties characterising repetitions in a word --\textit{period} and \textit{border}-- which play direct or indirect roles in several diverse applications ranging over pattern matching, text compression, assembly of genomic sequences and so on (see \cite{crochemore2007algorithms, crochemore2002jewels}). A period of a non-empty word $w$ of length $n$ is an integer $p$ such that $1 \leq p \leq n$, if $w[i] = w[i + p]$, for all $1 \leq i \leq n-p$. For instance, $3$, $6$, $7$, and $8$ are periods of the word $\texttt{aabaabaa}$. On the other hand, a border $u$ of $w$ is a (possibly empty) proper factor of $w$ occurring both as a prefix and as a suffix of $w$. For example, $\varepsilon$, $\texttt{a}$, $\texttt{aa}$, and $\texttt{aabaa}$ are the borders of $w=\texttt{aabaabaa}$. 

In fact, the notions of border and period are dual: the length of  each border of $w$ is equal to the length of $w$ minus the length of some period of $w$. For example,  \texttt{aa} is a border of the word \texttt{aabaabaa}; it corresponds to period $6=|\texttt{aabaabaa}|-|\texttt{aa}|$. Consequently, the basic data structure of periodicity on words is the \textit{border array} which stores the length of the longest border for each prefix of $w$. The computation of the border array of $w$ was the fundamental concept behind the first linear-time pattern matching algorithm -- given a word $w$ (pattern), find all its occurrences in a longer word $y$ (text). The border array of $w$ is better known as the ``failure function'' introduced in~\cite{morris1970linear} (see also~\cite{aho1987design}).  It is well-known that the border array of $w$ can be computed in $\mathcal{O}(n)$ time, where $n$ is the length of $w$, by a variant of the Knuth-Morris-Pratt algorithm~\cite{morris1970linear}.

Another notable aspect of the inter-dependency of these dual notions is the relationship between the length of the maximal unbordered factor of $w$ and the periodicity of $w$.  A~maximal unbordered factor is the longest factor of $w$ which does not have a border; its length is usually represented by $\mu(w)$, e.g. the maximal unbordered factor is  \texttt{aabab} and $\mu(w) = 5$ for the word $w=\texttt{baabab}$. This dependency has been a subject of interest in the literature for a long time, starting from the 1979 paper of Ehrenfeucht and Silberger~\cite{Ehrenfeucht} in which they raised the question -- at what length of $w$, $\mu(w)$ is maximal (i.e., equal to the minimal period of the word as it is well-known that it cannot be longer than that). This line of questioning, after being explored for more than three decades, culminated in 2012 with the work by Holub and Nowotka \cite{HOLUB2012668} where  an asymptotically optimal upper bound ($\mu(w)  \leq \frac{3}{7}n$) was presented; the historic overview of the related research can  be found in \cite{HOLUB2012668}.

Somewhat surprisingly, the symmetric computational problem---given a word $w$, compute the longest factor of $w$ that does not have a border---had not been studied until very recently. In 2015, Kucherov et al.~\cite{119f2fc2916f4ddd83626030d6641b36} considered this arguably natural problem and presented the first sub-quadratic-time solution.
A na\"{i}ve way to solve this problem
is to compute the border array starting at each position of $w$ and locating the rightmost zero, which results in an algorithm with $\mathcal{O}(n^2)$ worst-case running time.
On the other hand, the computation of the maximal unbordered factor can be done in linear time for the cases when $\mu(w)$ or its minimal period is small (i.e., at most half the length of $w$) using the linear-time computation of unbordered conjugates~\cite{DUVAL201477}. However, as has been illustrated in \cite{119f2fc2916f4ddd83626030d6641b36} and \cite{d67afd93990c4d2ea6906694e724fbb1}, most of the words do not fall in this category owing to the fact that they have large $\mu(w)$ and consequently large minimal period.  
In~\cite{119f2fc2916f4ddd83626030d6641b36}, an adaptation of the basic algorithm has been provided with average-case running time $\mathcal{O}(n^2/\sigma^4)$, where $\sigma$ is the alphabet's size; it has also been shown to work better, both in practice and asymptotically, than another straightforward approach that employs data structures from \cite{DBLP:conf/soda/KociumakaRRW15, 10.1007/978-3-642-34109-0_30} to query all relevant factors.

The currently fastest worst-case algorithm to compute the maximal unbordered factor of a given word takes $\mathcal{O}(n^{1.5})$ time; it was presented by Gawrychowski et al.~\cite{Gawrychowski2015} and it works for integer alphabets (alphabets of polynomial size in $n$). This algorithm works by categorising bordered factors into {\em short} borders and {\em long} borders depending on a threshold, and exploiting the fact that, for each position, the short borders are bounded by the threshold and the long borders are small in number. The resulting algorithm runs in $\mathcal{O}(n\log n)$ time on average. More recently, an $\mathcal{O}(n)$-time average-case algorithm was presented using a refined bound on the expected length of the maximal unbordered factor~\cite{d67afd93990c4d2ea6906694e724fbb1}.  

 \textbf{Our Contribution.} 
In this paper, we show how to efficiently answer the Longest Unbordered Factor question using combinatorial insight.
Specifically, we present an algorithm that computes the \textit{Longest Unbordered Factor Array} in $\Oh(n\log n)$ time with high probability. The algorithm can also be implemented deterministically in $\Oh(n \log n \log^2 \log n)$ time. This array specifies the length of the maximal unbordered factor  at each position in $w$. We thus improve on the running time of the currently fastest algorithm, which reports only the maximal unbordered factor of $w$ and works only for integer alphabets, taking $\mathcal{O}(n^{1.5})$ time. 
 
\textbf{Structure of the Paper.} 
In Section \ref{sect:max}, we present the preliminaries, some useful properties of unbordered words, the algorithmic toolbox, and a formal definition of the problem. We lay down the combinatorial foundation of the algorithm in Section \ref{sec:comb} and expound the algorithm in Section \ref{sec:algo}; its analysis is explicated in Section \ref{sec:analysis}. We conclude this paper with a final remark in Section \ref{sec:conclusion}.

\section{Background}\label{sect:max}

\noindent \textbf{Definitions and  Notation.}
We consider a finite \textit{alphabet} $\Sigma$ of \textit{letters}. Let $\Sigma^*$ be the set of all finite words over $\Sigma$.  The \textit{empty} word is denoted by $\varepsilon$. The \textit{length} of a word $w$ is denoted by $|w|$. For a word $w = w[1] w[2] \dd w[n]$,  $w[i\dd j]$ denotes the \textit{factor} $w[i]w[i+1]\dd w[j]$, where   $1 \leq i \leq j\leq n$. The {\em concatenation} of two words $u$ and $v$ is the word composed of the letters of $u$ followed by the letters of $v$. It is denoted by $uv$ or also by $u \cdot v$ to show the decomposition of the resulting word. Suppose $w=uv$, then $u$ is  a \textit{prefix}  and $v$ is  a \textit{suffix} of $w$; if $u\neq w$ then $u$ is a \textit{proper prefix} of $w$; similarly, if $v \neq w$ then $v$ is a \textit{proper suffix} of $w$. Throughout the paper we consider a non-empty word $w$ of length $n$ over a \textit{general alphabet} $\Sigma$; in this case, we replace each letter by its rank such that the resulting word consists of integers in the range $\{1,\ldots,n\}$. This can be done in $\mathcal{O}(n \log n)$ time after sorting the letters of $\Sigma$.


An integer $1\leq p \leq n$ is a \textit{period} of $w$ if and only if $w[i] = w[i+p]$ for all $1 \leq i \leq n-p$. The smallest period of $w$ is called the \textit{minimum period} (or \textit{the period}) of $w$, denoted by $\lambda(w)$. A  word $u$  ($u\neq w$) is  a \textit{border} of $w$, if $w = uv = v'u$ for some non-empty words $v$ and $v'$; note that $u$ is both a proper prefix and a suffix of $w$.   It should be clear that if $w$ has a border of length $|w|-p$ then it has a period $p$. Thus, the minimum period of $w$ corresponds to the length of the \textit{longest border} (or \textit{the border}) of $w$. Observe that the empty word $\varepsilon$ is a border of any word $w$.  If $u$ is the \textit{shortest border} then $u$ is the shortest {\it non-empty} border of $w$.

The word $w$ is called \textit{bordered} if it has a non-empty border, otherwise it is  \textit{unbordered}. Equivalently, the minimum period $p= |w|$ for an unbordered word $w$. Note that every bordered word $w$ has a shortest border $u$ such that $w = uvu$, where $u$ is unbordered. By $\mu(w)$ we denote the maximum length among all the unbordered factors of $w$.
\medskip

\noindent \textbf{Useful Properties of Unbordered Words.}
Recall that a word $u$ is a border of a word $w$ if and only if $u$ is both a  proper prefix and a  suffix of $w$. A border of a border of $w$ is also a border of $w$. A word $w$ is unbordered if and only if it has no non-empty border; equivalently $\varepsilon$ is the only border of $w$. The following properties related to unbordered words form the basis of our algorithm and were presented and proved in~\cite{Duval198231}.

\begin{proposition}[\cite{Duval198231}]
Let $w$ be a bordered word and $u$ be the shortest non-empty border of $w$. The following  propositions hold:
\begin{enumerate}
\item $u$ is an unbordered word;
\item $u$ is the unique unbordered prefix and suffix of $w$;
\item  $w$ has the form $w= uvu$.
\end{enumerate}
\end{proposition}

\begin{proposition}[\cite{Duval198231}]
\label{decompo} 
For any word $w$, there exists a unique sequence $(u_1,\cdots, u_k)$ of unbordered prefixes of $w$ such that $w= u_k\cdots u_1$. Furthermore, the following properties hold:   
\begin{enumerate}
\item $u_1$ is the shortest border of $w$;
\item $u_k$ is the longest unbordered prefix of $w$;
\item for all $i$, $1 \leq i\leq k$, $u_i$ is an unbordered prefix of $u_k$.
\end{enumerate}
\end{proposition}

The computation of the unique sequence  described  in Proposition~\ref{decompo} provides a unique   \textit{unbordered-decomposition} of a word.  
For instance, for $w = \texttt{baababbabab}$ the unique unbordered-decomposition of $w$ is
$\texttt{baa}\cdot\texttt{ba}\cdot \texttt{b}\cdot \texttt{ba}\cdot \texttt{ba}\cdot \texttt{b}$.
\medskip

\noindent \textbf{Longest Successor Factor (Length and Reference) Arrays.} Here, we present the arrays that will act as a toolbox for our algorithm.
The longest successor factor of $w$ (denoted by \textit{lsf}) starting at position $i$, is the longest factor of $w$ that occurs at $i$ and has at least one other occurrence in the suffix $w[i+1 \dd n]$. The {\sl longest successor factor array} gives for each position $i$ in $w$, the length of the longest factor starting both at position $i$ and at another position $j>i$. Formally, the longest successor factor array (\textsf{LSF}$_{\ell}$) is defined as follows.
\begingroup
\setlength\abovedisplayskip{3pt}
\setlength\belowdisplayskip{3pt}
\addtolength{\jot}{-0.3em}
\[
\textsf{LSF}_{\ell}[i] = \left\{ \begin{array}{ll}
0 &\mbox{ if~~} i = n, \\ 
 \mbox{max} \{ k \mid  w[i\dd i+k-1] = w[j\dd j+k-1  \},   & \mbox{  for ~} i < j \leq n. 
\end{array} \right.
\]
\endgroup

Additionally, we define the \textsf{LSF}-\textit{Reference Array}, denoted by \textsf{LSF}$_{r}$. This array specifies, for each position $i$ of $w$, the \textit{reference} of the longest successor factor at $i$. The \textit{reference} of $i$ is defined as the position $j$ of the last occurrence of  $w[i\dd i+ \textsf{LSF}_{\ell}[i]-1]$ in $w$; we say $i$ {\em refers to}  $j$. Formally, \textsf{LSF}-\textit{Reference Array} (\textsf{LSF}$_{r}$) is defined as follows.  
\begingroup
\setlength\abovedisplayskip{3pt}
\setlength\belowdisplayskip{3pt}
\addtolength{\jot}{-0.3em}
\[
\textsf{LSF}_{r}[i] = \left\{ \begin{array}{ll}
nil &\mbox{ if~~} \textsf{LSF}_{\ell}[i]=0, \\ 
 \mbox{max} \{ j \mid  w[j\dd j+\textsf{LSF}_{\ell}[i] -1 ] = w[i\dd i+\textsf{LSF}_{\ell}[i] -1 ]  \}   &    \mbox{  for ~} i < j \leq n. 
  
       \end{array} \right.
\]
\endgroup

\noindent {\em Computation:}
Note that the longest successor factor array is a mirror image of the well-studied longest previous factor array which can be computed in  $\mathcal{O}(n)$ time for integer alphabets~\cite{CI08,CIIKRW13}.
Moreover, in \cite{CI08}, an additional array that keeps a position of some previous occurrence of the longest previous factor was presented; such position may not be the leftmost. Arrays $\textsf{LSF}_{\ell}$ and $\textsf{LSF}_{r}$ can be computed using simple modifications (pertaining to the symmetry between the longest previous and successor factors) of this algorithm\footnote{The modified algorithm also computes some starting position $j>i$ for each factor $w[i\dd i+|\textsf{LSF}_\ell[i]|-1]$, $1\leq i\leq n$. Each such factor corresponds to the lowest common ancestor of the two terminal nodes in the suffix tree of $w$ representing the suffixes $w[i\dd n]$ and $j[j\dd n]$; this ancestor can be located in constant time after linear-time preprocessing~\cite{DBLP:conf/latin/BenderF00}. A linear-time preprocessing of the suffix tree also allows for constant-time computation of the rightmost starting position of each such factor.} 
within $\mathcal{O}(n)$ time for integer alphabets; see Appendix \ref{A:Example} for an example.

\begin{remark}
For brevity, we will use \textit{lsf} and \textit{luf} to represent the longest successor factor and the longest unbordered factor, respectively.
\end{remark}


\medskip
\noindent \textbf{Problem Definition.} The \textsc{Longest Unbordered Factor Array} problem can be  defined formally as follows. 

{\defproblem{\textsc{Longest Unbordered Factor Array}}{A word $w$ of length $n$.}{An array  \textsf{LUF}$[1\dd n]$ such that \textsf{LUF}$[i]$ is the length of the maximal unbordered factor starting at position $i$ in $w$, for all $1\leq i\leq n$.}}

\begingroup
\setlength\abovedisplayskip{0pt}
\setlength\belowdisplayskip{0pt}
\addtolength{\jot}{-0.3em}
\begin{example} \label{example}
Consider  $w =\texttt{aabbabaabbaababbabab}$, then the longest unbordered factor array is as follows. (\noindent Observe that $w$ is unbordered thus  $\mu(w) = |w| = 20$.)
\begin{table}[h]
\setlength\tabcolsep{3pt}
	\begin{center}
    \scalebox{1}{
        \begin{tabular}{|c|c|c|c|c|c|c|c|c|c|c|c|c|c|c|c|c|c|c|c|c|} \hline
~~$i~~$& 1 &  2&  3&  4&  5 &  6&  7&   8&  9&   10& 11& 12& 13 &14& 15& 16& 17& 18& 19& 20\\ \hline        

$w[i]$ &      \texttt{a~}&   \texttt{a~}&  \texttt{b~}&  \texttt{b~}&  \texttt{a~}&    \texttt{b~}&  \texttt{a~}&   \texttt{a~}&  \texttt{b~}&   \texttt{b~}&   \texttt{a~}&   \texttt{a~}&    \texttt{b~}&  \texttt{a~}&    \texttt{b~}&  \texttt{b~}&   \texttt{a~}&   \texttt{b~}&   \texttt{a~}&  \texttt{b~}\\ \hline \hline 

$\textsf{LUF[i]}$ &  20& 3&  12&  9& 12&  3& 14 & 3 &11 & 3&  10&  5&    2&  3  &  5&   2&   2&     2&  2&  1\\
\hline
\end{tabular}}
\end{center}
\end{table} \label{ex1} \vspace{-0.8cm}
\end{example}
\endgroup


\section{Combinatorial Tools}
\label{sec:comb}
The core of our algorithm  exploits the unique unbordered-decomposition of all suffixes of $w$ in order to compute the length of the maximal (longest) unbordered prefix of each such suffix. Let the unbordered-decomposition of $w[i\dd n]$ be $u_k\cdots u_1$ as in Proposition~\ref{decompo}. Then $\textsf{LUF}[i] = |u_k|$.   In order to compute the unbordered-decomposition for all the suffixes  {\em efficiently}, the algorithm uses the repetitive structure of $w$ characterised by the longest successor factor arrays. 

\noindent \textbf{Basis of the algorithm.} Abstractly, it is easy to observe that for  a given   position, if  the length of the longest successor factor  is zero (no factor starting at this position repeats afterwards) then the  suffix starting at that position is necessarily unbordered. On the other hand, if the length of the longest successor factor is smaller than the length of the unbordered factor at the reference (the  position of the the last occurrence of the longest successor factor) then  the ending positions of the longest unbordered factors at this position and that at its reference will coincide; these two cases are formalised in  Lemmas \ref{start} and \ref{small} below. The remaining case is not straightforward and its handling accounts for the bulk of the algorithm. 

\begin{lemma} \label{start}
If \textup{$\textsf{LSF}_{\ell}[i] =0$} then \textup{$\textsf{LUF}[i] = n-i+1$}, for $1 \leq i\leq n$. 
\end{lemma}

\begin{lemma} \label{small}
If \textup{$\textsf{LSF}_{r}[i] = j$} and \textup{$\textsf{LSF}_{\ell}[i] < \textsf{LUF}[j]$} then \textup{$\textsf{LUF}[i] = j + \textsf{LUF}[j] -i$}, for $1\leq i\leq n$. 
\end{lemma}
\begin{proof}
Let $k = j+\textsf{LUF}[j] -1$. 
  
We first show that $w[i\dd k]$ is unbordered.
Assume that   $w[i\dd k]$ is bordered and let $\beta$ be the length of one of its borders ($\beta < \textsf{LSF}_{\ell}[i]$ as $\textsf{LSF}_{r}[i] = j$). This implies that $w[i\dd i+\beta -1] = w[k-\beta +1\dd k]$. Since $w[i\dd i+\textsf{LSF}_{\ell}[i]-1] = w[j\dd j+ \textsf{LSF}_{\ell}[i]-1]$, we get $w[j\dd j+\beta -1] = w[k-\beta +1\dd k]$ (i.e., $w[j\dd k]$ is bordered) which is a contradiction. Moreover, $w[k+1 \dd n]$ can be factorised into prefixes of $w[j\dd k]$ (by definition of \textsf{LUF}); every such prefix is also a proper prefix of $w[i\dd i+\textsf{LSF}_{\ell}[i]-1]$ which will make every factor $w[i\dd k'],  k < k'\leq n$, to be bordered. This completes the proof.
\end{proof}

We introduce the notion of a \textit{hook} to handle finding the unbordered-decomposition of  suffixes $w[i\dd n]$ for the remaining case (i.e., when $\textsf{LSF}_{\ell}[i] \geq \textsf{LUF}[\textsf{LSF}_r[i]]$).

\begin{mydef}[Hook]
Consider a position $j$ in a length-$n$ word $w$.
Its \emph{hook} $\hook{j}{}$ is the smallest position $q$
such that $w[q\dd j-1]$ can be decomposed into unbordered prefixes of~$w[j\dd n]$.
\end{mydef}
The following observation 
provides a greedy construction of this decomposition.

\begin{observation}\label{obs:greedy}
The decomposition of a word $v$ into unbordered prefixes of another word $u$ is unique.
This decomposition can be constructed by iteratively trimming the shortest prefix of $u$ which occurs as a suffix
of the decomposed word.
\end{observation}
Moreover, the decomposability into unbordered prefixes of $u$ is hereditary in a certain sense:
\begin{observation}\label{obs:anyprefixes}
If a word $v$ can be decomposed into unbordered prefixes of $u$,
then every prefix of $v$ also admits such a decomposition.
\end{observation}

\begin{example} 
Consider  $w =\texttt{aabbabaabbaababbabab}$ as in Example~\ref{example}.
Observe that $\hook{18}{} = 13$: the factor $w[13\dd 17]= \texttt{ba}\cdot \texttt{b}\cdot \texttt{ba}$ can be decomposed into unbordered prefixes of $w[18\dd 20]=\texttt{bab}$. Moreover, no prefix of $w[18\dd 20]$ matches a suffix of $w[1\dd 12]=\cdots \texttt{aa}$.
\end{example}
The hook $\hook{j}{}$ has its utility when $j$ is a reference as shown in the following lemma.
\begin{lemma} \label{large}
Consider a position $i$ such that $\LSFl[i] \ge \LUF[j]$, where $j=\LSFr[i]$.
Then
\begingroup
\setlength\abovedisplayskip{0pt}
\setlength\belowdisplayskip{0pt}
\addtolength{\jot}{-0.3em}
\[
\LUF[i] =
\begin{cases}
\hook{j}{} - i & \text{if }i < \hook{j}{},\\
\LUF[j] & \text{otherwise.}
\end{cases}
\]
\endgroup
\end{lemma}

\begin{proof}
Let $u= w[j\dd j+\textsf{LUF}[j]-1]$, $v = w[i\dd i+\LSFl[i]-1]$, and $q = \hook{j}{}$.
Observe that $u$ occurs at position $i$ and that $w[q\dd n]$ can be decomposed into unbordered prefixes of $u$.

\noindent
\textbf{Case $a$: $i < q$.}
We shall prove that $w[i\dd q-1]$ is the longest unbordered prefix of $w[i\dd n]$; see~Fig.~\ref{fig:large}.
By~Observation~\ref{obs:anyprefixes}, any longer factor $w[i\dd k],~q\le k\le n$ has a suffix $w[q\dd k]$ composed of unbordered prefixes of $u$.
Thus, $w[i\dd k]$ must be bordered, because $u$ is its prefix.
To conclude, for a proof by contradiction suppose that $w[i\dd q-1]$ has a border $v'$.
Note that $|v'|\le \LSFl[i]$, so $v'$ is a prefix of $v$.
Hence, it occurs both as a suffix of $w[1\dd q-1]$ and a prefix of $w[j\dd n]$,
which contradicts the greedy construction of $q=\hook{j}{}$ (Observation~\ref{obs:greedy}).

\noindent
\textbf{Case $b$: $i\ge q$.}
The decomposition of $w[q\dd n]$ into unbordered prefixes of $u$ yields a decomposition
of $w[i\dd n]$ into unbordered prefixes of $u$, starting with $u$.
This is the unbordered-decomposition of $w[i\dd n]$ (see Proposition~\ref{decompo}),
which yields $\LUF[i]=|u|=\LUF[j]$.
\end{proof}

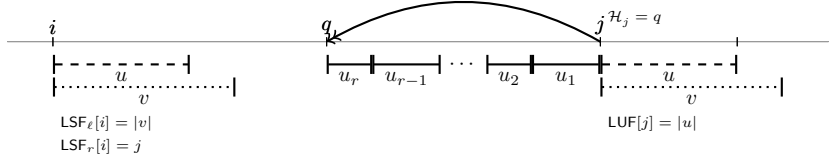
\begin{figure}[t]
\begin{center}
\begin{tikzpicture}[scale=0.6, every node/.style={scale=0.8}]
\draw [fill, ColorD8] (0,0) rectangle (18,0);

\foreach \x/\t in {1/i, 7/q, 13/j, 16/}	
  	\draw [black] (\x,-0.1) -- (\x, 0.1) node [above, midway] (\t) {$\t$}; 

\foreach \x/\y/\z/\t/\c in {1/5/-1/v/dotted, 13/17/-1/v/dotted, 13/16/-0.5/u/dashed,1/4/-0.5/u/dashed}
{
  	\draw [thick, \c, |-|] (\x,\z) --  (\y,\z) node [below, midway] {$\t$} ; 
};

\foreach \x/\t in {1/i, 7/q, 13/j, 16/}	
  	\draw [black] (\x,-0.1) -- (\x, 0.1) node [above, midway] {$\t$}; 
\foreach \x/\y/\t in {13/11.5/1, 11.5/10.5/2, 9.5/8/r-1}	
{
  	\draw [thick, |-|] (\x,-0.5) --  (\y,-0.5) node [below, midway] {$u_{\t}$} ; 
    
};
\node at (10,-0.5) {$\cdots$} ;
\draw [thick, |-|] (8,-0.5) --  (7,-0.5) node [below, midway] {$u_{r}$} ;

\node at (1,-1.8) [right] {\scriptsize{$\textsf{LSF}_\ell[i]=|v|$}};
\node at (1,-2.3) [right] {\scriptsize{$\textsf{LSF}_r[i]=j$}};
\node at (13,-1.8) [right] {\scriptsize{$\textsf{LUF}[j]=|u|$}};
\node at (13,0.5) [right] {\scriptsize{$\hook{j}{ }{}=q$}};
\draw [thick, ->>] (13,0) to [bend right]  (7,0) ;

\end{tikzpicture}
\end{center}
\caption{Case $a$ ($i < q$): The unbordered-decomposition of $w[i\dd n]$ consists of $w[i\dd  q-1]$ as the longest unbordered prefix, followed by a sequence of unbordered prefixes of $u$, including $u$ itself at position $j$. Therefore, $\LUF[i] = q-i$.} 
\label{fig:large}
\end{figure}


\section{Algorithm}
\label{sec:algo}
\alglanguage{pseudocode}

The algorithm operates in two  phases: a preprocessing phase followed by the main computation phase. The preprocessing  phase accomplishes the following: 
Firstly, compute the longest successor factor array $\textsf{LSF}_{\ell}$ together with   $\textsf{LSF}_{r}$ array. If $\textsf{LSF}_{r}[i]=j$ then we say $i$ {\em refers to}  $j$ and mark $j$ in a boolean array ($\textsf{IsReference}$) as a reference. 

In the main  phase, the algorithm computes the lengths of the longest unbordered factors for all positions in $w$. 
Moreover, it determines $\HOOK[j]=\hook{j}{}$ for each \emph{potential reference},
i.e., each position $j$ such that $j=\LSFr[i]$ and $\LSFl[i] \ge \LUF[j]$ for some $i<j$; see Lemma~\ref{large}.%

Positions are processed from right to left (in decreasing order) so that if $i$  refers to $j$, then $\textsf{LUF}[j]$ (and $\HOOK[j]$, if necessary) has already been computed before $i$ is considered. For each position $i$, the value of $\textsf{LUF}[i]$ is determined as follows:
\begin{enumerate}
\item  If $\textsf{LSF}_{\ell}[i] = 0$, then $\textsf{LUF}[i] = n -i +1$.

\item  Otherwise \begin{enumerate}
\item  If $\textsf{LSF}_{\ell}[i] <  \textsf{LUF}[j]$, then $\textsf{LUF}[i] = j+ \textsf{LUF}[j] -i$.

\item  If $\textsf{LSF}_{\ell}[i] \geq \textsf{LUF}[j]$ and $i \geq  \textsf{HOOK}[j]$,  then $\textsf{LUF}[i]= \textsf{LUF}[j]$.

\item If  $\textsf{LSF}_{\ell|}[i] \geq \textsf{LUF}[j]$ and $i < \textsf{HOOK}[j]$, then $\textsf{LUF}[i]= \textsf{HOOK}[j]-i$.
\end{enumerate}

\end{enumerate}

If $i$ is a potential reference, then $\HOOK[i]$ is also computed, as described in the Section~\ref{sec:hook} (see Algorithm~\ref{algo:1}  in Appendix~\ref{Appendix_Pseudo_Main} for a pseudo-code of the main algorithm). 
It is evident that the computational phase of the algorithm fundamentally reduces to finding the hooks for potential references; for brevity, the term reference will mean a potential reference  hereafter.


\subsection{Finding Hook (\textsc{FindHook} Function)}\label{sec:hook}

\subparagraph{Main idea} When  \textsc{FindHook} is called on a reference $j$, it must return $\hook{j}{}$. 
A simple greedy approach follows directly from Observation~\ref{obs:greedy}; see also Figure~\ref{fig:hook}.
Initially, the factor $w[1\dd j-1]$ is considered and the shortest suffix of $w[1\dd j-1]$ which is a prefix of $w[j\dd n]$ is computed. Then this suffix, denoted $u_{1}=w[i_1\dd j-1]$,
is truncated (chopped) from the considered factor $w[1\dd j-1]$; the next factor considered will be $w[1\dd i_1-1]$. In general, we  iteratively compute and  truncate  the shortest prefixes of $w[j\dd n]$ from the right end of the considered factor; shortening the length of the considered factor in each iteration and terminating as soon as no prefix of $w[j\dd n]$ can be found. If the considered factor at termination is $w[1\dd q-1]$, position $q$ is returned by the function as $\hook{j}{}$. 

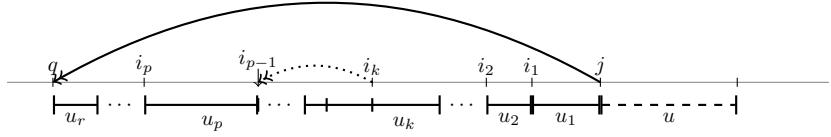
\begin{figure}[t]
\begin{center}
\vspace{-1em}
\begin{tikzpicture}[scale=0.6, every node/.style={scale=0.8}]
\draw [fill, ColorD8] (0,0) rectangle (18,0);

\foreach \x/\t in {1/q, 13/j, 16/}	
  	\draw [black] (\x,-0.1) -- (\x, 0.1) node [above, midway] (\t) {$\t$}; 

\foreach \x/\y/\z/\t/\c in {13/16/-0.5/u/dashed}
{
  	\draw [thick, \c, |-|] (\x,\z) --  (\y,\z) node [below, midway] {$\t$} ; 
};

\foreach \x/\y/\t in {13/11.5/1, 11.5/10.5/2}	
{
  	\draw [thick, |-|] (\x,-0.5) --  (\y,-0.5) node [below, midway] {$u_{\t}$} ; 
    \draw [black] (\y,-0.1) -- (\y, 0.1) node [above, midway] {$i_{\t}$};
};
\foreach \x/\y in {9.5/6.5, 5.5/3}	
{
  	\draw [thick, |-|] (\x,-0.5) --  (\y,-0.5) node [below, midway] {} ;  
};
\draw [thick, |-|] (2,-0.5) --  (1,-0.5) node [below, midway] {$u_{r}$} ; 
\foreach \y in {8, 7,5.5}	
{
    \draw [black, thick] (\y,-0.35) -- (\y, -0.65) node [above, midway] {};
};
\foreach \x/\t in {8/k, 3/p}	
{
    \draw [black] (\x,-0.1) -- (\x, 0.1) node [above, midway] {$i_{\t}$};
};
\foreach \x/\t in {8.7/k, 4.5/p}	
{
    \node at (\x,-0.9) {$u_{\t}$} ;
};
\draw [black, ->] (5.5,0.3) -- (5.5, -0.1) node [above, midway] {$i_{p-1}$};

\node at (10,-0.5) {$\cdots$} ;
\node at (6,-0.5) {$\cdots$} ;
\node at (2.5,-0.5) {$\cdots$} ;

\draw [thick, ->>] (13,0) to [bend right]  (1,0) ;
\draw [thick, dotted, ->>] (8,0) to [bend right]  (5.5,0) ;

\end{tikzpicture}
\end{center}

\caption{A chain of consecutive shortest prefixes of $w[j\dd n]$ starting at positions $i_1 > i_2 > \cdots > i_r=q$. No prefix of $w[j\dd n]$ is a suffix of $w[1\dd q-1]$, so the hook value of position $j$ is $\hook{j}{} = q$. Meanwhile, $\HOOK[i_k]$ is set to $i_{p-1}$ in order to avoid iterating through $i_{k+1},\ldots,i_{p-1}$ again.}  
\label{fig:hook}
\end{figure}

The factors $w[q\dd j-1]$ considered by successive calls of \textsc{FindHook} function may overlap.  Moreover, the same \textit{chains} of consecutive unbordered prefixes may be computed several times throughout the algorithm. To expedite the chain computation in the subsequent calls of \textsc{FindHook} on another reference $j'$ ($j' < j$), we can \textit{recycle} some of the computations done for  $j$ by shifting the value $\HOOK[\cdot]$ of each such index (at which a  prefix was cut for $j$) leftwards (towards its final value). Consider the starting position $i_k$ at which $u_{k}$ was cut (i.e., $u_{k} = w[i_k\dd i_{k-1}-1]$ is the shortest unbordered prefix of $w[j\dd n]$ computed at $i_{k-1}$). Let $i_p$ be the first position considered after $i_k$ such that $|u_{p}| > |u_{k}|$. In this case, every factor $u_{k+1},\ldots,u_{p-1}$ is a prefix of $u_{k}$; see Figure~\ref{fig:hook}. Therefore, $w[i_{p-1}\dd i_{k}-1]$
can be decomposed into prefixes of $u_k$ (and of $w[i_k\dd n]$).
Consequently, we set $\textsf{HOOK}[i_k]= i_{p-1}$ so that the next time a prefix of length greater than  or equal to $|u_{k}|$ is cut at $i_k$, we do not have to repeat truncating the prefixes $u_{k+1},\ldots,u_{p-1}$ and we may start directly from position $i_{p-1}$.

In order to express the intermediate values in the $\HOOK$ table, we generalize the notion of $\hook{j}{}$:
for a position $j$ and a length $\ell$, we define $\hook{j}{\ell}$ as the largest position $q$ such that $w[q\dd j-1]$
can be decomposed into unbordered prefixes of $w[j\dd n]$ whose lengths do not exceed $\ell$. 
Observe that $\hook{j}{0}=j$ and $\hook{j}{\ell} = \hook{j}{}$ if $\ell \ge \LUF[j]$.
 
\subparagraph{Implementation} For each position $i_k$, we set $\HOOK[i_k] = \hook{i_k}{|u_k|}$, equal to $i_{p-1}$ in the case considered above.
Computing these values for all indices $i_k$ can be efficiently realised using a stack. Every starting position $i_p$, at which $u_{p}$ is cut, is pushed onto the stack as a (length, position) pair $(|u_{p}|, i_p)$. Before pushing, every element $(|u_k|, i_k)$ such that $|u_{k}| < |u_{p}|$ is popped and the hook value of index  $i_k$ is updated ($\textsf{HOOK}[i_k] = \hook{i_k}{|u_k|}=i_{p-1} = i_p + |u_p|$). A pseudo-code implementation of this function is given below as Algorithm~\ref{algo:2}.

\begin{algorithm}[h]
\alglanguage{pseudocode}
\begin{algorithmic}[1]

\Function{$\textsc{FindHook}$}{$j$}
\State $\texttt{st} \leftarrow \textsc{NewStack}()$
\State $q \leftarrow \textsf{HOOK}[j]$
\State $\beta \leftarrow  \textsc{FindBeta}(q, j)$ \Comment{\mbox{\scriptsize{the length of the shortest prefix of $w[j\dd n]$ ending at $q-1$}}}
\While{$(\beta \neq  0)$}
\State $\textsc{HandlePopping}(\texttt{st}, j, q, \beta)$
\State $\textsc{Push}(\texttt{st}, (\beta,q-\beta))$
\State $q \leftarrow \textsf{HOOK}[q-\beta]$
\State $\beta \leftarrow \textsc{FindBeta}(q, j)$
\EndWhile

\State $\textsc{HandlePopping}(\texttt{st}, j, q, \infty)$

\State \Return $q$ \Comment{\mbox{\scriptsize{returns $\hook{j}{}$ }}}
\EndFunction
\medskip
\Function{$\textsc{HandlePopping}$}{$\texttt{st}, j, q, \beta$}
\While{$\textsc{IsNotEmpty}(\texttt{st})$ and $\textsc{Length}(\textsc{Top}(\texttt{st})) < \beta$} 
 \State $(\texttt{length,pos}) \leftarrow \textsc{Pop}(\texttt{st})$
 \State $\textsf{HOOK}[\texttt{pos}] \leftarrow q$\Comment{\mbox{\scriptsize{$q = \hook{\texttt{pos}}{\texttt{length}}$}}}

 \EndWhile

\EndFunction
\end{algorithmic}
\caption{Return $\hook{j}{}$ and set $\textsf{HOOK}[i] \leftarrow \hook{i}{\beta}$ for each $(\beta,i)$ pushed onto the stack}
\label{algo:2}
\end{algorithm}

\subparagraph{Analysis} Throughout the algorithm, each unbordered prefix $u_{p}$  at  position $i_p$ is computed just once by the \textsc{FindHook} function. Nevertheless, a longer\footnote{It will be easy to deduce after Lemma \ref{lemma_childtwinset} that the length of the prefix cut (the next time) at the same position will be at least twice the length of the current prefix cut at it.}
 unbordered prefix $u'_{p}$  may  be computed at $i_p$ again when \textsc{FindHook}  is called on reference $j'$ (where $q < j' < j$).

In what follows, we introduce certain characteristics of the computed unbordered prefixes which aids in establishing the relationship between the stacks of various references. Let $\mathcal{S}_j$ be the set of positions pushed onto the stack during a call of  $\textsc{FindHook}$ on reference $j$.

\begin{mydef}[Twin Set]
A \textit{twin set} of reference $j$ for length $\ell$, denoted by $\mathcal{T}_j^{\ell}$, is the set of all the positions $i\in \mathcal{S}_j$ which were pushed  onto the stack paired with length $\ell$  in the call of \textsf{FindHook} on reference $j$  (i.e., $\mathcal{T}_j^{\ell} = \{ i ~|~ (\ell, i) \text{ was pushed onto the stack of }j\}$). 
\end{mydef}
Note that a {\em unique} shortest unbordered prefix of $w[j\dd \textsf{LUF}[j]-1]$ occurs at each $i$ belonging to the same twin set. However, as and when a longer prefix at $i$ is cut (say $\ell'$) for another reference $j' < j$, $i$ will be added to $\mathcal{T}_{j'}^{\ell'}$. 
\begin{remark}
$\mathcal{S}_{j} = \bigcup \limits_{\ell=1}^{\textsf{LUF}[j]} \mathcal{T}_{j}^{\ell}$. 
\end{remark}

Hereafter, a twin set will essentially imply a non-empty twin set.

\begin{lemma}
\label{lemma_maxhook}
If $j'$ and $j$ are references such that $j' \in \Stack_{j}$, then $\hook{j}{} \le  \hook{j'}{}$.
\end{lemma}

\begin{proof}
Since $j'\in \Stack_j$, the suffix $w[j'\dd n]$ (and, by Observation~\ref{obs:anyprefixes}, its every prefix $w[j'\dd k]$) can be decomposed into unbordered prefixes of $w[j\dd n]$.
Consequently, any decomposition into unbordered prefixes of $w[j'\dd n]$ yields a decomposition into unbordered prefixes of $w[j\dd n]$. 
In particular, $w[\hook{j'}{}\dd n]$ admits such a decomposition, which implies $\hook{j}{} \le \hook{j'}{}$.
\end{proof}

If the stack $\mathcal{S}_j$ is the most recent stack containing a reference $j'$,
we say that $j'$ is the parent of $j$. More formally, the parent of $j'$ is defined as $\min\{j : j' \in \mathcal{S}_j\}$.
If $j'$ does not belong to any stack (and thus has no parent), we will call it a \textit{\textbf{base reference}}.

\begin{lemma}
\label{lemma_childtwinset}
If $j$ and $j'$ are two references such that $j$ is the parent of $j'$ and $j' \in \mathcal{T}_{j}^{\ell}$, then 
each position $i\in \Stack_{j'}$ satisfies the following properties:
\begin{enumerate}
\item \label{l1} $i\in \mathcal{T}_{j}^{\ell}$;
\item \label{l2} there exists $k\in \mathcal{T}_{j}^{\ell'}$, with $\ell' > \ell$, such that $(k+\ell'-i, i)$ is pushed onto the stack of $j'$.
\end{enumerate}
\end{lemma}

\begin{figure}[t!]
\begin{center}
\begin{tikzpicture}[yscale=0.6,xscale=.8, every node/.style={scale=0.8}]
\draw [fill, ColorD8] (0,0) rectangle (15,0);

\draw [black] (1,-0.1) -- (1, 0.1) node [above, midway] (q) {$\hook{j}{}$}; 
\draw [black] (13,-0.1) -- (13, 0.1) node [above, midway] (j) {$j$}; 

\foreach \x/\y/\z/\t/\c in {3/5.5/-1.5/z/solid}
{
  	\draw [thick, \c, |-|] (\x,\z) --  (\y,\z) node [below, midway] {$\t$} ; 
};

\foreach \x/\y/\z/\t/\n in {7/9.5/-1.5/z/ }
{
  	\draw [thick, dotted, |-|] (\x,\z) --  (\y,\z) node [below, midway] {$\t$} node [right] {\n} ; 
};

\foreach \x/\y/\t in {13/11.5/1, 11.5/11/2}	
{
  	\draw [thick, |-|] (\x,-0.5) --  (\y,-0.5) node [below, midway] {} ; 
};
\foreach \x/\y in {10/8.5, 7.5/7, 5.5/4.5, 3.5/3, 2/1}	
{
  	\draw [thick, |-|] (\x,-0.5) --  (\y,-0.5) node [below, midway] {} ;  
};

\foreach \x/\y in {8.5/7.5, 4.7/3.5}	
{
    \draw [black, thick, dashed] (\x,-0.5) -- (\y, -0.5) node [above, midway] {};
};
\foreach \x/\t in {7/j', 5.5/p, 4.5/k, 3/i}	
{
    \draw [black] (\x,-0.1) -- (\x, 0.1) node [above, midway] {$\t$};
};
\foreach \x/\t in {7.25/v, 5/v', 3.3/v}	
{
    \node at (\x,-0.9) {${\t}$} ;
};
\foreach \x in {10.5, 6.26, 2.5 }	
{
   \node at (\x,-0.5) {$\cdots$} ;
};

\draw [thick, ->>] (13,0) to [bend right]  (1,0) ;
\draw [thick, dotted, ->>] (7,0) to [bend right]  (5.5,0) ;
\end{tikzpicture}
\end{center}
\caption{The pair $(|z|,i)$ is the first to be pushed onto the stack of   $j'$. The factor $z$ is unbordered,  has $v$ as a proper prefix and some $v'$ as a proper suffix,  where both $v$ and $v'$ are unbordered prefixes  of $w[j\dd n]$  whose lengths $\ell$ and $\ell'$, respectively, satisfy $\ell<\ell'$.}
\label{fig:references}
\end{figure}
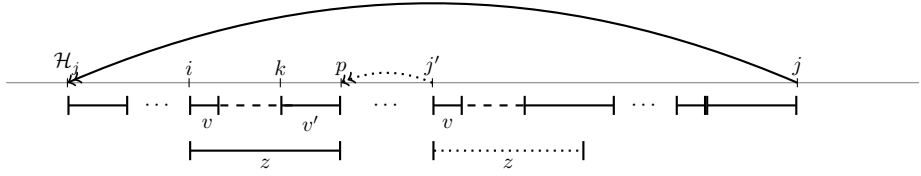

\begin{proof}
Let $p$ be the value of $\HOOK[j']$ prior to the execution of \textsc{FindHook}$(j')$. 
Since $j'\in \Twin_j^\ell$, the earlier call \textsc{FindHook}$(j)$ has set $\HOOK[j']=\hook{j'}{\ell}$.
As $j$ is the parent of $j'$, no further call has updated $\HOOK[j']$.
Thus, we conclude that $p = \hook{j'}{\ell}$.

Consequently, the first pair pushed onto the stack of $j'$ (cf.~Algorithm~\ref{algo:2}) is $(|z|,i)$,
where $z= w[i\dd p-1]$ is the shortest suffix of $w[1\dd p-1]$ which also occurs as a prefix of $w[j'\dd n]$ (see Figure~\ref{fig:references}).
Moreover, observe that  $|z|>\ell$ by the greedy construction of $\hook{j'}{\ell}$.

Recall that $j'\in \Twin_j^\ell$ implies that $w[j'\dd n]$ can be decomposed into unbordered prefixes of $w[j\dd n]$,
with the first prefix of length $\ell$, denoted $v = w[j'\dd j'+\ell-1]$.
With an occurrence at position $j'$, the factor $z$ also admits such a decomposition, still with the first prefix $v$ (due to $|z|>|v|$).
Additionally, note that $w[p\dd j'-1]$ can be decomposed into unbordered prefixes of $v$. 
Concatenating the decompositions of $z=w[i\dd p-1]$, $w[p\dd j'-1]$, and $w[j'\dd n]$, we conclude that $w[i\dd n]$
can be decomposed into unbordered prefixes of $w[j\dd n]$ with the first prefix (in this unique decomposition) equal to $v$.
Hence, $i\in \Stack_{j'}$ belongs to the same twin set as $j'$; i.e., it satisfies the first claim of the lemma.

Additionally, in the aforementioned decomposition of $w[i\dd n]$ consider the factor $v' = w[k\dd p-1]$ which ends at position $p-1$.
By the greedy construction of $\hook{j'}{\ell}$, its length $|v'|$ is strictly larger than $\ell$,
so $k\in \Twin^{\ell'}_{j}$ for $\ell'=|v'|>\ell$. 
Moreover, recall that $(|z|,i)=(k+\ell'-i,i)$ is pushed onto the stack of $j'$.
Consequently, $i$ also satisfies the second claim of the lemma.

A similar reasoning is valid for each $i$ that will appear in $\mathcal{S}_{j'}$. 
\end{proof}


\begin{lemma}
\label{lemma_children}
If $j$ is the parent of two references $j'' < j'$, both of which belong to $\mathcal{T}_{j}^{\ell}$, then $\mathcal{S}_{j'} \cap \mathcal{S}_{j''} = \emptyset$.
\end{lemma}
\begin{proof}
Let $u =w[j\dd j+ \textup{\textsf{LUF}}[j] -1]$ and $v$ be the shortest unbordered prefix of $u$ cut at $j'$ and $j''$ (i.e., $|v| = \ell$).  Let $u' = w[j'\dd j'+ \textup{\textsf{LUF}}[j'] -1]$ and $u'' = w[j''\dd j''+ \textup{\textsf{LUF}}[j''] -1]$. Here, the current call to the \textsc{FindHook} function has been made on the reference $j''$. Consider a position $i$ such that  $i \in \mathcal{S}_{j'}$ and $i$ would also appear in $\mathcal{S}_{j''}$; let the corresponding prefixes of $u'$ and $u''$ cut at $i$ be $z'$ and $z''$  (examine Figure~\ref{fig:children}).  Observe that $i$ was in $\mathcal{T}_{j}^{\ell}$ (Lemma \ref{lemma_childtwinset}),  therefore, each of $z'$ and $z''$ has $v$ as a proper prefix. Let $v'$ and $v''$ be the corresponding proper suffixes of $z'$ and $z''$ where  $v'$ and $v''$ are unbordered prefixes of $u$; both of length greater than $|v|$.

\begin{figure}[!h]
\begin{center}
\begin{tikzpicture}[scale=0.6, every node/.style={scale=0.8}]
\draw [fill, ColorD8] (0,0) rectangle (18,0);
\foreach \x/\t in {0.5/q, 15/j, 17/}	
  	\draw [black] (\x,-0.1) -- (\x, 0.1) node [above, midway] (\t) {$\t$}; 

\foreach \x/\y/\z/\t/\c in {15/17/-0.5/u}
{
  	\draw [decorate,decoration={zigzag}, thick, |-|] (\x,\z) --  (\y,\z) node [below, midway] {$\t$} ; 
};
\foreach \x/\y/\z/\t/\c in {2.5/5/-1.5/z'/solid, 2.5/8.5/-2.1/z''/solid}
{
  	\draw [thick, \c, |-|] (\x,\z) --  (\y,\z) node [below, midway] {$\t$} ; 
};

\foreach \x/\y/\z/\t/\n in {10.5/17.5/-2.1/u''/, 13.5/17.5/-1.5/u'/}
{
  	\draw [thick, dotted, |-|] (\x,\z) --  (\y,\z) node [below, midway] {$\t$} node [right] {\n} ; 
};

\foreach \x/\y/\t in {10.5/11/1, 13.5/14/2}	
{
  	\draw [thick, |-|] (\x,-0.5) --  (\y,-0.5) node [below, midway] {} ; 
};
\foreach \x/\y in {8.5/7, 6.5/6, 5/4, 3/2.5, 1.5/0.5}	
{
  	\draw [thick, |-|] (\x,-0.5) --  (\y,-0.5) node [below, midway] {} ;  
};

\foreach \x/\y in {7/6.5, 4.2/3}	
{
    \draw [black, thick, dashed] (\x,-0.5) -- (\y, -0.5) node [above, midway] {};
};
\foreach \x/\t in {10.5/j'', 5/p, 8.5/k, 2.5/i, 13.5/j'}	
{
    \draw [black] (\x,-0.1) -- (\x, 0.1) node [above, midway] {$\t$};
};
\foreach \x/\t in {6.25/v, 4.5/v', 2.8/v, 7.75/v'', 13.75/v, 10.75/v}	
{
    \node at (\x,-0.9) {${\t}$} ;
};
\foreach \x in {12.25, 9.5, 5.5, 2 }	
{
   \node at (\x,-0.5) {$\cdots$} ;
};
\draw[decoration={brace,raise=5pt},decorate]
  (8.5,-0.8) -- node[below=10pt] {$x$} (6,-0.8);

\draw [thick, ->>] (15,0) to [bend right]  (0.5,0);

\end{tikzpicture}
\end{center}
\caption{The pair $(|z'|,i)$ and $(|z''|,i)$ are pushed onto the stack of $j'$ and $j''$, respectively, where $i$ is a position common to both $\mathcal{S}_{j'}$ and $\mathcal{S}_{j''}$. 
}
\label{fig:children}
\end{figure}
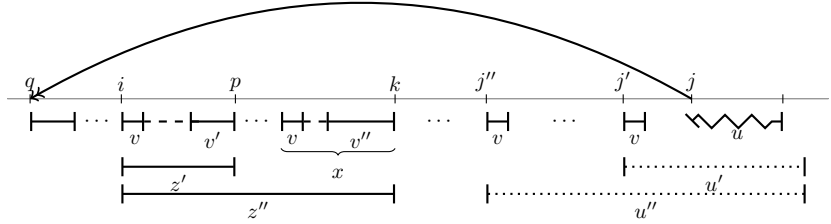

Since  $i\in \mathcal{S}_{j'}$, $w[i\dd j'-1]$   can be decomposed from right to left into unbordered prefixes of $u'$ such that each prefix (say $\tilde{v}$) of $u$ having length greater than $|v|$ that had been computed when $j$ was considered, is \textit{covered}; i.e., $\tilde{v}$ appears as either a proper suffix or a factor of some shortest prefix in such decomposition. In other words, the shortest prefix of $u'$ that ends with $\tilde{v}$ starts from the nearest $v$ preceding $\tilde{v}$ whose corresponding position was pushed in $\mathcal{S}_{j'}$. Note that the same condition is also valid for any prefix of $\tilde{v}$ that is longer than $v$.

\begin{enumerate}
\item \textbf{$|z'| < |z''|$: } The factor  $w[p\dd k-1]$ can be decomposed into unbordered prefixes of $u'$, where $p$ and $k$ are as in Figure~\ref{fig:children}. Let $x$ be the rightmost   prefix of such decomposition ($x$ has $v$ as proper prefix,   $v''$ as a proper suffix, and the corresponding position of this $v$; i.e., $i+|z''|-|x|$ is in $\mathcal{S}_{j'}$). Moreover, $|z'| < |x|$ otherwise  $z''$ cannot be  unbordered. Observe the two  equal-length factors $w[i\dd i+|z''|]$  and $w[j''\dd j''+|z''|]$. It follows from $i+|z''|-|x| \in \mathcal{S}_{j'}$ that
$j''+|z''|-|x| \in \mathcal{S}_{j'}$. Consequently, $w[i\dd i+|z''|-|x|]$  and $w[j''\dd j''+|z''|-|x|]$ have the same right to left decomposition in prefixes of $u'$ implying that if if $i\in \mathcal{S}_{j'}$ then $j''$ should have been in $\mathcal{S}_{j'}$ which is a contradiction.

\item \textbf{$|z'| \geq |z''|$: }  This implies that  $z''$ is  an unbordered prefix of $u'$.  Therefore,  if $i$ was pushed onto the stack of $j'$ then $j''$ should also had been pushed onto its stack  which is a contradiction.\qedhere
\end{enumerate}
\end{proof}

\subsection{Finding Shortest Border (\textsc{FindBeta} Function)}
Given a reference $j$  and a position $q$, function \textsc{FindBeta} returns the length $\beta$ of the shortest prefix of $w[j\dd n]$ that is a suffix  of $w[1\dd q-1]$,
or $\beta = 0$ if there is no such prefix; note that the sought shortest prefix is necessarily unbordered.

To find this length, we use `prefix-suffix queries' of \cite{DBLP:conf/soda/KociumakaRRW15,10.1007/978-3-642-34109-0_30}.
Such a query, given a positive integer $d$ and two factors $x$ and $y$ of $w$, reports all prefixes of $x$ of length between $d$ and $2d$ that occur as suffixes of $y$. 
The lengths of sought prefixes are represented as an arithmetic progression, which makes it trivial to extract the smallest one.
A single prefix-suffix query can be implemented in $\Oh(1)$ time after randomized preprocessing of $w$ which takes $\Oh(n)$ time in expectation~\cite{DBLP:conf/soda/KociumakaRRW15},
or $\Oh(n \log n)$ time with high probability~\cite{10.1007/978-3-642-34109-0_30}. Additionally, replacing hash tables with deterministic dictionaries~\cite{DBLP:conf/icalp/Ruzic08},
yields an $\Oh(n\log n \log^2 \log n)$-time deterministic preprocessing.

To implement \textsc{FindBeta}, we set $x=[j\dd n]$, $y=[1\dd q-1]$ and we ask prefix-suffix queries for subsequent values $d=1,3,\ldots,2^{k}-1,\ldots$
until $d$ exceeds $\min(|x|,|y|)$. Note that we can terminate the search as soon as a query reports a non-empty answer.
Hence, the running time is $\Oh(1+\log \beta)$ if the query is successful (i.e., $\beta \ne 0$) and $\Oh(\log n)$ otherwise.

Furthermore, we can expedite the successful calls to \textsc{FindBeta} if we already know that $\beta\notin\{1,\ldots,\ell\}$.
In this case, 
we can start the search with $d = \ell+1$. 
Specifically, if $j$ is not a base reference and belongs to $\mathcal{T}_{j'}^{\ell}$ for some $j'$, we can start from $d = 2\ell+1$ because  Lemma~\ref{lemma_childtwinset}.2 guarantees that $\beta > \ell+\ell'\ge 2\ell$.

%
%
\newcommand{\FindBeta}{\ensuremath{\textsc{FindBeta}}\xspace}
\newcommand{\FindHook}{\ensuremath{\textsc{FindHook}}\xspace}

\section{Analysis}
\label{sec:analysis}
Algorithm~\ref{algo:1}  computes the longest unbordered factor at  each position $i$; position $i$ is a start-reference or it refers to some other position. The correctness of the computed  $\textsf{LUF}[i]$ follows directly from Lemmas~\ref{start} through~\ref{large}. 

The analysis of the algorithm running time  necessitates probing of the total time consumed by \textsc{FindHook} and the time spent by \textsc{FindBeta} function which, in turn, can be measured in terms of the total  size of the stacks of various references.  

\begin{lemma} \label{lemma_totalsize}
The total size of all the stacks used throughout the algorithm is $\Oh(n \log n)$. 
Moreover, the total running time of the \FindBeta function is $\Oh(n\log n)$.
\end{lemma}

\begin{proof}%
First, we shall prove that any position $p$ belongs to $\Oh(\log n)$ stacks.

By Lemma~\ref{lemma_childtwinset}.1, the stack of any reference is a subset of the stack of its parent.
Moreover, by Lemmas~\ref{lemma_childtwinset}.1 and~\ref{lemma_children}, the stacks of references sharing the same parent are disjoint.
A similar argument (see Lemma~\ref{disjoined} in Appendix \ref{A:lemma_base_disjoint}) shows that the stacks of base references are disjoint.

Consequently, the references $j_1>\ldots>j_s$ whose stacks $\mathcal{S}_{j_i}$ contain $p$ form a chain with respect to the parent relation:
$j_1$ is a base reference, and the parent of any subsequent $j_i$ is $j_{i-1}$.
Let us define $\ell_1,\ldots,\ell_s$ so that $p \in \mathcal{T}_{j_i}^{\ell_i}$. 
By Lemma~\ref{lemma_childtwinset}.2, for each $1\le i < s$, there exist $k_i$ and $\ell'_i>\ell_i$ such that
$k_i \in \mathcal{T}_{j_i}^{\ell'_i}$ and $\ell_{i+1} = k_i-p +\ell'_i \ge \ell_i + \ell'_i > 2\ell_i$.
Due to $1\le \ell_i \le n$, this yields $s \le 1+\log n = \Oh(\log n)$, as claimed.

Next, let us analyse the successful calls $\beta = \FindBeta(q,j)$ with $p=q-\beta$.
Observe that after each such call, $p$ is inserted to the stack $\mathcal{S}_j$ and to the twin set $\mathcal{T}_{j}^\beta$,
i.e, $j=j_i$ and $\beta = \ell_i$ for some $1\le i \le s$.
Moreover, if $i>0$, then $j_{i} \in \mathcal{T}_{j_{i-1}}^{\ell_{i-1}}$,
which we are aware of while calling  \FindBeta.
Hence, we can make use of the fact that $\ell_i \notin \{1,\ldots,2\ell_{i-1}\}$ to find $\beta = \ell_i$
in time $\Oh(\log \frac{\ell_{i}}{\ell_{i-1}})$. For $i=1$, the running time is $\Oh(1+\log \ell_1)$.
Hence, the overall running time of successful queries $\beta = \FindBeta(q,j)$ with $p=q-\beta$
is $\Oh(1+\log \ell_1 + \sum_{i=2}^s \log \frac{\ell_{i}}{\ell_{i-1}})=\Oh(1+\log \ell_s)= \Oh(\log n)$,
which sums up to $\Oh(n\log n)$ across all positions $p$. 

As far as the unsuccessful calls $0=\FindBeta(q,j)$ are concerned,
we observe that each such call terminates the enclosing execution of \FindHook.
Hence, the number of such calls is bounded by $n$ and their overall running time is clearly $\Oh(n \log n)$.
\end{proof}

\begin{theorem} \label{theorem_final}
Given a word $w$ of length $n$, Algorithm~\ref{algo:1} solves the \textsc{Longest Unbordered Factor Array} problem in $\Oh(n\log n)$ time
with high probability. It can also be implemented deterministically in $\Oh(n \log n \log^2 \log n)$ time. 
\end{theorem}
\begin{proof}
Assuming an integer alphabet, the  computation of $\textsf{LSF}_{\ell}$ and $\textsf{LSF}_{r}$ arrays  along with the  constant time per position initialisation of the other  arrays sum up  the preprocessing stage  (Lines 2--7) to  $\mathcal{O}(n)$ time. The running time required for the assignment of the \textsf{luf} for all  positions (Lines 9--18) is  $\mathcal{O}(n)$. The time spent in construction of the data structure to answer prefix-suffix queries used in  \textsc{FindBeta} function is $\mathcal{O}(n \log n)$ with high probability or $\Oh(n \log n \log^2 \log n)$ deterministic.

Additionally, the total running time  of  the \textsc{FindHook} function  for all the references, being proportional to the aggregate size of all the stacks, can be deduced from Lemma~\ref{lemma_totalsize}. This has been shown to be   $\mathcal{O}(n\log n)$ in the worst case, same as the total running time of \FindBeta. The claimed bound on the overall running time follows.
\end{proof}
To show that the upper bound shown in Lemma~\ref{lemma_totalsize} 
in the worst case is tight, we design an infinite family of words that exhibit the worst-case behaviour (see Appendix~\ref{Appendix_Worst} for more details). 
  
%

\section{Final Remark}
\label{sec:conclusion}

Computing the longest unbordered factor in $o(n\log n)$ time for integer alphabets remains an open question.



\bibliographystyle{plain}
\bibliography{unbordered}
\newpage
\appendix
\section{Appendix}
\subsection{Pseudo-code} \label{Appendix_Pseudo_Main}
\alglanguage{pseudocode}
\begin{algorithm}[H]
\footnotesize
\caption{Compute the Longest Unbordered Factor Array}
\label{Algorithm:overabundant}
\begin{algorithmic}[1]
\Procedure{$\textsc{LongestUnborderedFactor}$}{word $w$, integer $n=|w|$}
\medskip
 \Statex \hskip0.3em $\triangleright$ {\bf Preprocessing:}
    \State $\textsf{LSF}_{\ell}, \textsf{LSF}_r \leftarrow$ \textsc{LongestSuccessorFactor}($w$)
    \For{$i \leftarrow 1 $ {\bf to} $n$  }
		\If{$\textsf{LSF}_{\ell}[i] \neq 0$}		
            \State $\textsf{IsReference}[\textsf{LSF}_r[i]]\leftarrow true$
		\EndIf   	
    \EndFor
    \State $\textsf{HOOK}[1\dd n] \leftarrow 1, \cdots, n$

 \medskip
 
\Statex \hskip0.3em $\triangleright$ {\bf Main Algorithm:}
     \For{$i\leftarrow n$ to 1 }
     	\If{$\textsf{LSF}_{\ell}[i] =0$} \Comment{\mbox{\scriptsize{starting reference}}}		
            \State $\textsf{LUF}[i] \leftarrow n-i+1$
		\Else 
     		\State $j\leftarrow \textsf{LSF}_{r}[i]$
     		\If{$\textsf{LSF}_{\ell}[i] <  \textsf{LUF}[j]$} 
						\State   $\textsf{LUF}[i] \leftarrow  j+ \textsf{LUF}[j] -i$
			\ElsIf{$i \geq \textsf{HOOK}[j]$}
						\State $\textsf{LUF}[i] \leftarrow \textsf{LUF}[j]$ \label{update}
			\Else \State $\textsf{LUF}[i]\leftarrow  \textsf{HOOK}[j]-i$
			\EndIf	
     
            \If{\color{black}\textsc{IsPotentialReference}($i$)   \color{black}} \footnotemark  
                \State $\textsf{HOOK}[{i}]	\leftarrow$ \textsc{FindHook}($i$)
         \EndIf
   \EndIf
  \EndFor
\EndProcedure
\Statex
\end{algorithmic}

  \label{algo:1}
\end{algorithm}
\footnotetext{ \color{black}\textsc{IsPotentialReference}($i$)  returns {\it true}  if there exists $i'$ such that $\textsf{LSF}_{r}[i'] =i$   and $\textsf{LSF}_{\ell}[i'] \geq \textsf{LUF}[i]$.
\color{black}}

A \texttt{C++} implementation of our algorithm can also be made available upon request.

\subsection{Words Exhibiting Worst-Case Behaviour}\label{Appendix_Worst}
A word can be made to exhibit the worst-case behaviour if we force the maximum number of  positions to be pushed onto $\Theta(\log n)$  stacks. This can be achieved as follows.
\begin{enumerate}
\item Maximise the number of references: Every position in each twin set $\mathcal{T}_{j}^{l}$ is a reference.
\item Maximise the size of each stack: The largest position (reference) in any twin set pushes the rest of the positions onto its stack. If $j'$ is largest reference in $\mathcal{T}_{j}^{\ell}$ then $\mathcal{S}_{j'} = \mathcal{T}_{j}^{\ell} - \{j'\}$.
\item Maximise the number of twin sets obtained for a stack: This increases the number of unbordered prefixes that can be cut at some position $i$, therefore, increasing the number of re-pushes of $i$. This can be achieved by keeping   $|\mathcal{T}_{j}^{\ell}| = 2|\mathcal{T}_{j}^{\ell+1}|+1$.
\end{enumerate}
Using the above, Algorithm~\ref{algo:3} creates a word $w$ over $\Sigma= \{
\texttt{a,b}\}$, such that  the total size of the stack of the  base reference $j$ ($w[j] = \texttt{a}$) and the references  that appear in $\mathcal{S}_{j}$ is $t|\mathcal{S}_{j}|-t$, where $t$ is the number of non-empty twin sets obtained for the stack of $j$.

Assume  a binary alphabet $\Sigma = \{\texttt{a}, \texttt{b}\}$, the following  words exhibit the maximum total  size of the stacks used: $w_3 = (\texttt{aabaabb})^2, t_{max}= 3$; $w_4 = (\texttt{aabaabbaabaabbb})^2, t_{max}=4$; $w_5 = (\texttt{aabaabbaabaabbbaabaabbaabaabbbb})^2, t_{max}=5$; etc., where $t_{max}$ is the maximum number of stacks onto which some proportional number of elements has been pushed by  Algorithm~\ref{algo:1}.   Position 1 in $w_4$, for example, is pushed onto four  stacks paired with length 1, 3, 7 then 15. The total size of the stacks used by each word from this family of words is thus $\Theta(n \log n)$. Figure~\ref{figure} (in Appendix~\ref{Appendix_Worst}) shows the relation between the lengths of the words and the total size of the stacks used by Algorithm \ref{algo:1} for the specified family of the words.
\begin{algorithm}
\alglanguage{pseudocode}
\caption{Create Word $w$ Over $\Sigma= \{
\texttt{a,b}\}$}
\footnotesize
\begin{algorithmic}[1]
\State $w \leftarrow ``"$
\State $\texttt{block} \leftarrow ``\texttt{a}"$
\State $i \leftarrow 1$
\While{$i \leq  t-1$}
\State $w\leftarrow w + \texttt{block} + w$
\State $\texttt{block} \leftarrow \texttt{block} + ``\texttt{b}"$
\State $i \leftarrow i+1$
\EndWhile
\State $w \leftarrow w + \texttt{block}$
\State $w\leftarrow w+ w$
\end{algorithmic}
\label{algo:3}
\end{algorithm}
\begin{figure}[!h]
\centering
\begin{tikzpicture}[scale=1]
\begin{axis}[
  width = 0.9\linewidth,
  height= 6cm,
  xlabel= $n / 10^5$,
  axis y line*=left,
  ylabel=$t_{max}$]
\addplot table [y=$depth-max$,   x expr=\thisrow{n}/100000   ]{data.dat};  \label{t}
\end{axis}

\begin{axis}[
  width = 0.9\linewidth,  
  height= 6cm,
  axis y line*=right,
  axis x line=none,
  ylabel=total size of stacks $/10^6$,
  legend pos=south east
]
\addlegendimage{/pgfplots/refstyle=t}\addlegendentry{$t_{max}$}
\addplot[smooth,mark=x,red] table [y expr =\thisrow{$np$}/1000000, x=n]{data.dat};

\addlegendentry{total size of stacks}
\end{axis}
\end{tikzpicture}
\caption{Plot showing the maximum number of stacks onto which some element has been pushed ($t_{max}$) and the total size of stacks  used by Algorithm~\ref{algo:1} for  specially designed  words.} 

\label{figure}
\end{figure}
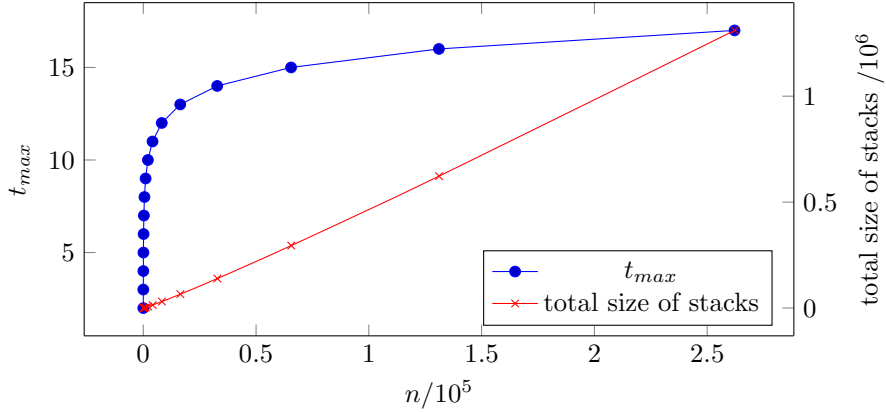

\subsection{Stacks of the Base References are Disjoint}\label{A:lemma_base_disjoint}
\begin{lemma}\label{disjoined}
If $j_1\ne j_2$ are base references, then $\mathcal{S}_{j_1}\cap \mathcal{S}_{j_2}=\emptyset$.
\end{lemma}
\renewcommand{\S}{\mathcal{S}}

\begin{proof}
 Suppose that the \textsc{FindHook} function is called for each position in $w$. We define a  \emph{base position} analogously as  a  position that does not appear in any stack. 
 For a proof by contradiction, let $i$ be the largest element of $\mathcal{S}_{j_1}\cap \mathcal{S}_{j_2}$,
 with $(\ell_1,i)$ and  $(\ell_2,i)$  pushed  onto the stacks of $j_1$ and $j_2$, respectively. 
Note that $i+\ell_1 \in \{j_1\}\cup \S_{j_1}$ and $i+\ell_2 \in \{j_2\}\cup \S_{j_2}$.
Thus, our choice of $j_1\ne j_2$ as base positions and $i$ as the largest element of $\mathcal{S}_{j_1}\cap \mathcal{S}_{j_2}$
guarantees $\ell_1 \ne \ell_2$.
We assume that $\ell_1 < \ell_2$ without loss of generality.

Let $u$ be the longest unbordered factor at $j_1$. 
Note that due to $i\in \S_{j_1}$, the suffix $w[i\dd n]$ can be decomposed into unbordered prefixes of $u$.
In particular, $w[i\dd i+\ell_2-1]$ admits such a decomposition $w[i\dd i+\ell_2-1] = v_1\cdots v_r$ with $|v_1|=\ell_1$.
Moreover, observe that $|v_r|>\ell_1$; otherwise, $v_r$ would be a border of $w[i\dd i+\ell_2-1]$.

Let $v_s$ be the first of these factors satisfying $|v_s|>\ell_1$ and let $k=j_1+|v_1\cdots v_{s-1}|$.
Note that $w[j_1\dd k-1]$ admits a decomposition $w[j_1\dd k-1]=v_1\cdots v_{s-1}$ into unbordered prefixes of $v_s$.
Consequently,  $j_1\in \S_{k}$ if $k$ is a base position, and $j_1\in \S_{k'}$ if $k$ is not a base position and $k\in \S_{k'}$ for some base position $k'$.
In either case, this contradicts the assumption that $j_1$ is a base position.

In fact, Algorithm~\ref{algo:1} calls the \textsc{FindHook} function on a subset of positions; i.e., potential references.
However, as we show below, all base references are actually base positions. 
For a proof by contradiction, suppose that $j'$ is a base reference, but it would have been pushed onto the stack of a base position $j>j'$.

Below, we show that the longest unbordered factor at $j$, denoted $u$, does not have any other occurrence in $w$.
First, suppose that it occurs at a position $k>j$.
Observe that $w[j\dd k-1]$ can be decomposed into unbordered prefixes of $u$.
Consequently,  $j \in \S_{k}$ if $k$ is a base position, and $j\in \S_{k'}$ if $k$ is not a base position and $k\in \S_{k'}$ for some base position $k'$.
In either case, this contradicts the assumption that $j$ is a base position.
Next, suppose that $u$ occurs at a position $k<j$ and let us choose the largest such $k$.
Observe that $\textsf{LSF}_\ell[k] \ge |u|$ and $\textsf{LSF}_r[k]=j$ since $j$ is the only position of $u$ larger than $k$.
However, this means that $j$ is a potential reference, contrary to our assumption.

In particular, we conclude that $u'= w[j' \dd j+|u|-1]$ does not have any border of length $|u|$ or more.
On the other hand, shorter borders are excluded since $u$ is unbordered and $u'$ can be decomposed into unbordered prefixes of $u$.
Consequently, $u'$ is unbordered. However, $j'$ is a potential reference, so $u'$ occurs to the left of $j'$.
This yields an occurrence of $u$ to the left of $j$, a contradiction.
\end{proof}
\subsection{Example of Longest Successor Factor Arrays}\label{A:Example}
\begin{example} 
Let   $w =\texttt{aabbabaabbaababbabab}$. The associated arrays are as follows.
\begin{table}[h]
\setlength\tabcolsep{3pt}
	\begin{center}
     \scalebox{0.9
     }{
        \begin{tabular}{|c|c|c|c|c|c|c|c|c|c|c|c|c|c|c|c|c|c|c|c|c|} \hline
~~$i~~$& 1 &  2&  3&  4&  5 &  6&  7&   8&  9&   10& 11& 12& 13 &14& 15& 16& 17& 18& 19& 20\\ \hline        

$w[i]$ &      \texttt{a~}&   \texttt{a~}&  \texttt{b~}&  \texttt{b~}&  \texttt{a~}&    \texttt{b~}&  \texttt{a~}&   \texttt{a~}&  \texttt{b~}&   \texttt{b~}&   \texttt{a~}&   \texttt{a~}&    \texttt{b~}&  \texttt{a~}&    \texttt{b~}&  \texttt{b~}&   \texttt{a~}&   \texttt{b~}&   \texttt{a~}&  \texttt{b~}\\ \hline \hline 

$\textsf{LSF}_{\ell}[i]$ &  5& 6&  5& 4 &3&4& 3&  4& 3 & 2 &1 & 4&  3  &  2&   1&   3&     2&  1& 0& 0\\ \hline
$\textsf{LSF}_{r}[i]$ &  7& 14&  15&  16& 17&  10& 11 & 14 &15 & 18&  19&  17&    18&  19  &  20&   18&   19&     20&  nil&  nil\\
\hline

\end{tabular}}
\end{center}
\end{table} \label{ex2}
\vspace{-0.8cm}
\end{example}

\end{document}